\documentclass[aps,prl,twocolumn,longbibliography,superscriptaddress]{revtex4-2}

\usepackage{amsmath, empheq,braket,amsthm,amssymb,graphics,amsfonts,array,color,subfigure, mathrsfs}
\usepackage[unicode=true,pdfusetitle, bookmarks=true,bookmarksnumbered=false,bookmarksopen=false,breaklinks=false,pdfborder={0 0 0},backref=false,colorlinks=true,linkcolor=blue,citecolor=blue,urlcolor=blue]{hyperref}

\usepackage{graphicx,subfigure,lipsum}
\newcommand{\tr}[1]{\mathrm{Tr}\left[ #1 \right]}

\newtheorem{theorem}{Theorem}
\newtheorem{theo}{Theorem}
\newtheorem{lem}{Lemma}

\newtheorem{cor}{Corollary}
\newtheorem{corollary}{Corollary}
\newtheorem{prop}{Proposition}
\newtheorem{proposition}{Proposition}

\begin{document}


\title{Native linear-optical protocol for efficient multivariate trace estimation}


\author{Leonardo Novo}
\email{leonardo.novo@inl.int}
\affiliation{International Iberian Nanotechnology Laboratory (INL), Av. Mestre Jos\'e Veiga, 4715-330 Braga, Portugal}

\author{Marco Robbio}
\affiliation{International Iberian Nanotechnology Laboratory (INL), Av. Mestre Jos\'e Veiga, 4715-330 Braga, Portugal}
\affiliation{Centre for Quantum Information and Communication, \'Ecole polytechnique de Bruxelles, CP 165/59, Universit\'e libre de Bruxelles, 1050 Brussels, Belgium}

\author{Ernesto F. Galvão}
\affiliation{International Iberian Nanotechnology Laboratory (INL), Av. Mestre Jos\'e Veiga, 4715-330 Braga, Portugal}
\affiliation{Instituto de F\'isica, Universidade Federal Fluminense, Av. Gal. Milton Tavares de Souza s/n, Niter\'oi, RJ, 24210-340, Brazil}

\author{Nicolas J. Cerf}
\affiliation{Centre for Quantum Information and Communication, \'Ecole polytechnique de Bruxelles, CP 165/59, Universit\'e libre de Bruxelles, 1050 Brussels, Belgium}


\begin{abstract}
The Hong-Ou-Mandel test estimates the overlap between spectral functions characterizing the internal degrees of freedom of two single photons. It can be viewed as a photon-native protocol that implements the well-known quantum SWAP test. Here, we propose a native linear-optical protocol that efficiently estimates multivariate traces of quantum states called Bargmann invariants, which are ubiquitous in quantum mechanics. Our protocol may be understood as a photon-native version of the cycle test in the circuit model, which encompasses many-photon multimode quantum states. We show the protocol is sample-efficient and discuss applications, such as generalized suppression laws, efficient quantum kernel estimation for quantum machine learning, eigenspectrum estimation, and the characterization of multiphoton indistinguishability.
\end{abstract}

\maketitle

While photons are ideal carriers of quantum information due to their long coherence times, a major obstacle in building a universal photonic quantum computer is the fact that they do not naturally interact. Indeed, implementing two-qubit gates in photonics is costly and requires either probabilistic schemes \cite{Knill2001, Bartolucci2023}, strong optical non-linearities, or hybrid schemes involving light-matter interaction \cite{lindner2009proposal, huet2025deterministic}. It is thus natural to ask what kind of photon-native quantum algorithms \cite{salavrakos2025photon_native} or useful quantum computation tasks can be implemented deterministically using purely linear optics.

An important example of such photon-native quantum protocols is the SWAP test, which  estimates the overlap between two quantum states. It is well known that, via the Hong-Ou-Mandel (HOM) effect \cite{hongMeasurementSubpicosecondTime1987}, the overlap between two wavefunctions describing the internal degrees of freedom of a single photon (such as its frequency, time delay, polarization, or orbital angular momentum) can be estimated by measuring the coincidence rate after the two photons interfere at an unbiased beam splitter. While this equivalence between the SWAP and HOM tests \cite{HOM_swap} yields a photon-native technique to estimate the overlap between two single-photon wavefunctions in any dimension (i.e., two qudits), it does not provide a method to measure overlaps of arbitrary states of light featuring multiple photons occupying several modes of the electromagnetic field. This is in stark contrast with the circuit model, as the SWAP test can also efficiently estimate overlaps between two $n$-qubit states in an exponentially large Hilbert space, 
leaving a gap between the computational power of the circuit model and that of linear optics for the problem of overlap estimation.

An important generalization of the SWAP test, sometimes called the cycle test \cite{ekert2002direct,ekert2003direct,brun2004}, enables the estimation of the multivariate trace of a tuple of $M$ quantum states $\{\rho_{1}, \ldots, \rho_{M}\}$, given by $\tr{\rho_1 \rho_2 ...\rho_M }$. These quantities, also referred to as Bargmann invariants of order $M$, remain unchanged under the action of the same unitary on every state. As such, they capture structural basis-independent properties of quantum states and measurements~\cite{Oszmaniec2024, wagner2024quantum, Fernandes_Barg, arvidsson24}. For example, their phase corresponds to the Pancharatnam geometric phase~\cite{pancharatnam1956generalized} associated with a sequence of projective measurements, or equivalently, the Berry phase along a cyclic geodesic path~\cite{simon1993bargmann, chruscinski2004geometric}. Moreover, they  have been studied in relation to photonic indistinguishability~\cite{Menssen_17,Menssen_22,annoniIncoherentBehaviorPartially2025}, Kirkwood–Dirac quasi-probability distributions~\cite{bamber2014observing,kirkwood1933quantum}, and more recently applied to topics such as error mitigation~\cite{liang2023unified}, weak values~\cite{wagner2023simple} and out-of-time-ordered correlators~\cite{wagner2024quantum}. 
While the cycle test \cite{ekert2002direct,ekert2003direct,brun2004,Oszmaniec2024} and its smaller depth variants \cite{quek2024multivariate} can efficiently evaluate multivariate traces of multi-qubit quantum states with quantum circuits, an efficient linear-optical protocol for this task is still missing. In fact, existing protocols based on linear optics \cite{Pont_22, wu2022sparse} apply only to single-photon quantum states and have exponentially large sample complexity.

In this work, we close this important gap by introducing a linear-optical protocol based on Fourier interferometry and photo-counting measurements that estimates multivariate traces of general quantum states of light -- possibly with many photons in multiple modes -- without requiring postselection or auxiliary photons.  The complexity of the protocol in terms of the number of optical components depends on the type of encoding used. To give an example, if the $\rho_i$'s are $M$ states of $n$ dual-rail qubits, the protocol requires $2n$ Fourier interferometers, each involving $M$ spatial modes, plus a mesh of $poly(M)$ photonic SWAP operations. 
As another important result, we generalize the suppression laws for Fourier interferometers -- which were previously known for input Fock states -- to arbitrary (many-photon multimode) states of light. Moreover, we show how violations of these generalized suppression laws are the key to estimate multivariate traces of photonic quantum states.

The expected applications are manifold, but we limit ourselves to discussing  some of the most impactful ones, such as efficient Quantum Kernel reconstruction for Quantum Machine Learning, estimation of state spectra, and the characterization of photonic indistinguishability.

\paragraph{Framework --} 
We consider $M$ bosonic quantum states $\{\rho_{1}, \ldots, \rho_{M}\}$, where each $\rho_{j}$ belongs to identical Fock spaces $\mathcal{H}$ associated to multiple bosonic modes (spatial, time-bin, frequency, polarization, etc.).  The corresponding bosonic creation operators are denoted as $a^\dagger_{\alpha}$, where $\alpha$ could also possibly be a continuous parameter. Considering now the larger Hilbert space made by the (Fock) tensor product of $M$ identical Hilbert spaces $\mathcal{H}$, we define the associated creation operators $a^\dagger_{j,\alpha}$, where $j\in \{1,...,M\}$, as well as the initial state 
\begin{equation}\label{eq:in_state}
    \Omega= \rho_1 \otimes \rho_2 \otimes\dots\otimes \rho_M.  
\end{equation}
The degrees of freedom (d.o.f.) labeled by $\alpha$ will be referred to as the internal d.o.f. of each quantum state $\rho_{j}$, whereas the label $j$ refers to the Hilbert spaces associated to the different quantum systems $\rho_j$. 

Our protocol for multivariate trace estimation requires a linear interference process between these systems $\rho_j$, while leaving their internal d.o.f. invariant. Such an interference process is described by a $M\times M$ unitary matrix $U$ acting on the creation operators as 
\begin{equation}\label{eq:linear_interference}
   \hat{U} \hat{a}^{\dagger}_{j, \alpha}   \hat{U}^\dagger= \sum_{k=1}^M U_{k,j} \, \hat{a}^{\dagger}_{k, \alpha}, \forall j, \alpha. 
\end{equation}
Note that the number of physical elements to perform this transformation depends on the encoding of the internal d.o.f.  For example, in the case where $\rho_1$ and $\rho_2$ are two single photons in arbitrary polarization states ($\alpha\in \{H,V\}$), then a single beam splitter (which preserves polarization) can be used to implement this operator. In turn, if $\rho_1$ and $\rho_2$ are two dual-rail photonic qubits, two identical beam splitters
would be needed -- one for interference of the two rails corresponding to the ``$0$" state of each qubit and another for interfering the two rails corresponding to state ``$1$".    

Following the interference process, we consider a particle counting measurement which ignores internal d.o.f. described by  $\hat{N}_j = \sum_\alpha \hat{n}_{j, \alpha}$, where $\hat{n}_{j, \alpha}= \hat{a}^{\dagger}_{j, \alpha}\hat{a}_{j, \alpha}$. Similarly to the examples given above, the number of physical detectors depends on the encoding of the information. In the first scenario, a single polarization-independent photocounter per system would suffice. In the second scenario, each dual-rail qubit would be measured by two detectors (one per rail) and the the act of ignoring the internal d.o.f. $\alpha$ (whether the photons occupy the first or second rail) would be done by classical postprocessing.  The (coarse-grained) outcome associated to each observable $\hat{N}_j$ is noted $S_j$. Then, we denote the joint outcome when measuring these $M$ observables as a pattern vector $\vec{S}= (S_1, \dots, S_M)$ and the corresponding probability distribution as $D_{\vec{S}}$. 
The generalized HOM test between two many-photon multimode states $\rho_1$ and $\rho_2$, which will be introduced later in this work, is an example of the interference and measurement process described above (see Fig.~\ref{fig:generalizedHOM}).

\paragraph{Multivariate trace estimation ---} 
In the circuit model, the cycle test enables multivariate trace estimation \cite{ekert2002direct,ekert2003direct,brun2004, Oszmaniec2024} but requires a cyclic permutation between different quantum states, controlled by an auxiliary qubit, in a circuit known as a Hadamard test. In linear optics, such controlled operations are not readily available. Instead, we use techniques developed in Ref. \cite{daley2012measuring} to enable the measurement of the entanglement dynamics of bosonic atoms in an optical lattice, based on the estimation of traces of powers of a quantum state $\tr{\rho^n}$ by linearly interfering $n$~copies of $\rho$. Such a multicopy technique was also proven to give access to more complex nonlinear functionals of $\rho$, providing linear optical interferometric schemes in order to probe optical uncertainty, nonclassicality, or entanglement \cite{multicopy-cerf-1,multicopy-cerf-2,multicopy-cerf-3,multicopy-cerf-4}. Here, instead of considering $n$~identical copies, we consider a general scenario with an initial state of $M$~different bosonic systems as described in Eq.~\eqref{eq:in_state}.  The cyclic permutation $\hat{C}$ between these systems is a special type of linear interferometer as in \eqref{eq:linear_interference} since it acts on the creation operators as
\begin{equation}
    \hat{C}a^{\dagger}_{j,\alpha}\hat{C}^\dagger = a^{\dagger}_{j+1,\alpha}  \, , \quad \forall j,\alpha ,  
\end{equation}
where the index $j$ is understood modulo $M$. This operator can be be diagonalized by the Fourier interferometer, namely $\hat{C}= \hat{F}\hat{D}\hat{F}^{\dagger}$, with 
\begin{align}\label{eq: Fourier evolution}
    \hat{F}a^{\dagger}_{j,\alpha}\hat{F}^\dagger = \frac{1}{\sqrt{M}}\sum_{k=0}^{M-1}\omega^{jk} a^{\dagger}_{k,\alpha} \, , \quad
    \hat{D}a^{\dagger}_{j,\alpha}\hat{D}^\dagger = \omega^j a^{\dagger}_{j,\alpha} , 
\end{align}
and $\omega=\exp(2\pi i/M)$.
Of course, any diagonal linear interferometer $\hat{D}$ can be decomposed as a set of phase shifters acting in parallel and can be written as 
\begin{equation}
\hat{D}= \prod_{j=0}^{M-1}\exp\left(i \frac{2\pi j }{M}  \hat{N}_j\right)   \,  ,
\end{equation}
which implies that the cyclic shift operator $\hat{C}$ becomes diagonal in the particle number basis after applying the Fourier interferometer (in a sense, the unitary $\hat{C}$ can be viewed as a complex-valued observable). This diagonalization allows us to estimate a multivariate trace with photo-counting and appropriate classical postprocessing of the outcomes $(S_j)$ associated to observable $\hat{N}_j$. More precisely, we have 
\begin{align}\label{eq: multivariate_cycle}
  \Delta_{1..M}&= \tr{\rho_1\rho_2\dots \rho_M}\nonumber\\ 
  &=\tr{\hat{C}\, \rho_1\otimes \rho_2\otimes \dots\otimes \rho_{M}} \nonumber \\
&=\tr{\hat{D}~\Omega_\text{out}}
\end{align}  
where $\Omega_\text{out}$ is the output state obtained at the output of an inverse Fourier interferometer, that is,
\begin{equation}\label{eq:output_Fourier}
 \Omega_\text{out}=\hat{F}^\dagger \rho_1\otimes\rho_2\otimes..\otimes\rho_M  \, \hat{F}.
\end{equation}
Thus, $\Delta_{1\dots M}$ is accessible as the expectation value of $\hat{D}$ in the output state $\Omega_\text{out}$. The proof of the first step of Eq.~\eqref{eq: multivariate_cycle} is given in the Appendix. We stress that the sole assumption in our protocol is that each state $\rho_j$ belongs to the same Hilbert space $\mathcal{H}$, so that $\rho_j$ can in principle contain an arbitrary number of particles or a superposition of different particle numbers.  
This versatility also allows us to estimate multivariate traces of photonic quantum states encoded in different d.o.f. Each $\rho_j$ could, for example, be a highly entangled state prepared by a universal photonic quantum computer. Furthermore, the protocol is efficient in the sense that an additive error on the estimate of any  $\Delta_{1\dots M}$ can be obtained by sampling the outcomes of our linear optical circuit, where the number of samples scales polynomially in the inverse error but is agnostic about the number of states $M$. More precisely, the following proposition can be derived by using Hoeffding's inequality (see Appendix). 
\begin{prop}
    Let $\Omega= \bigotimes_{j} \rho_j$ be a bosonic quantum state where each  $\rho_j$ belongs to a multimode bosonic Fock space $\mathcal{H}$. The multivariate trace $\tr{\rho_1\rho_2...\rho_M}$ can be estimated with probability $1-\delta$ and precision $\epsilon$ with  $O(\epsilon^{-2}\ln\delta^{-1})$ samples. \rm{(See proof in the Appendix.)}
\end{prop}
As previously mentioned, the number of optical elements to implement the Fourier interferometer according to Eq.~\eqref{eq: Fourier evolution} depends on the physical encoding of the internal d.o.f. $\alpha$. For $M$ dual-rail qubits, this would require two Fourier interferometers of $M$ spatial modes each, with additional permutation of spatial modes so that the 1st rails of all qubits interfere together in the same Fourier interferometer, and similarly for the 2nd rails. More generally, if each $\rho_j$ has $d$ spatial modes and additional internal modes (e.g. frequency or polarization), the protocol would require $d$ Fourier interferometers. 
Known efficient decompositions of multimode interferometers into beam splitters and phase shifters \cite{reckExperimentalRealizationAny1994,clementsOptimalDesignUniversal2017} guarantee that our protocol requires a number of optical elements that is polynomial in $M$ and $d$.

\paragraph{Generalized suppression laws ---} 
An alternative but very useful way to understand how Fourier interferometry allows for multivariate trace estimation involves recognizing a connection with suppression laws, or zero-transmission laws, i.e., the existence of some forbidden input-output transitions in optical interferometers with certain symmetries \cite{tichy2010zero, crespi15, viggianiello18, dittel2018totally}. The simplest such suppression law is the HOM effect, where two indistinguishable photons entering opposite arms of an unbiased beam-splitter always bunch together at the output \cite{hongMeasurementSubpicosecondTime1987}. In the HOM experiment, it is well known that this suppression law is violated as a consequence of imperfect particle indistinguishability. 
The overlap between the internal states of the two interfering single photons can be estimated from the bunching probability (i.e., the complement of probability of the ideally suppressed coincidence outcome $\vec{S}=(1,1)$), namely  
\begin{align}
P_{(0,2)}+ P_{(2,0)}=1-P_{(1,1)}=  \frac{1+\tr{\rho_1 \rho_2}}{2} \label{eq:P_bunch_HOM}.
\end{align}

While suppression laws have been formulated for Fourier \cite{tichy2010zero} and Sylvester \cite{crespi15, viggianiello18} interferometers, a detailed understanding of what information can be extracted when a suppression law is violated has been missing. In this context, our main contributions are twofold. First, we generalize suppression laws for Fourier interferometers involving  different multimode many-photon states $\rho_j$. This goes beyond previous works that considered only Fock-state inputs \cite{tichy2010zero, dittel2018totally}.  In addition,  we explain how violations of suppression laws allow the estimation of multivariate traces of the input states.   

Before we present our results, it is useful to define the function
\begin{equation}\label{eq:f_of_S}
f(\vec{S})= \sum_{j=0}^{M-1} j S_j ~~(\text{mod}~M ), 
\end{equation}
which is known to give rise to a zero-transmission law: it was proven in \cite{tichy2010zero, dittel2018totally} that when sending $M$ indistinguishable single photons in a Fourier interferometer, any outcome $\vec{S}$ such that $f(\vec{S})\neq 0$ is suppressed. Moreover, we define the probabilities 
\begin{equation}\label{eq:P_j}
    P_j = \underset{\vec{S}\sim D_{\vec{S}}}{\mathrm{Prob}}[f(\vec{S})=j], \quad j\in \{0,..., M-1\}
\end{equation}
and the expectation values of the unitary operator $\hat{C}^k$ on the input state of a Fourier interferometer
\begin{equation}\label{eq:X_k}
    X_k = \tr{\hat{C}^k \Omega}, \quad k\in \{0,..., M-1\}. 
\end{equation}
Our main result is the following (proof in Appendix~\ref{sec: Proofs}). 
\begin{theo}\label{theo:violation_suplaws}
    The probabilities $P_j = \mathbb{E}[\delta_{f(\vec{S}),j}]$ are related to the expectation values $X_k$ via a discrete Fourier transform, namely
    \begin{equation}
        P_j = \frac{1}{M} \sum_{k=0}^{M-1}\omega^{-jk} X_k, \label{eq:Pj_vs_Xk}
    \end{equation}
with $\omega=\exp(2\pi i/M)$. 
\end{theo}
This result can be understood by re-expressing the expectation values as $X_k = \mathbb{E}[ \omega^{k f(\vec{S})}]$.
Note that $f(\vec{S})$ can be seen as a way to bin the exponentially large space of possible outcomes $\vec{S}$ into $M$ possible bins, each with aggregate probability $P_j$. Hence, by estimating these probabilities 
from experimental samples, the values of $X_k$ can be estimated by an inverse discrete Fourier transform of these values. Note that $X_1=\Delta_{1...M} $ is the multivariate trace of the $M$ input states. 
In addition, this method allows us to obtain extra relational information about the input states via the other values of $X_k$ (except for $X_0$ which is always equal to $1$). To give an example for four states, we have that $X_2= \tr{\rho_1\rho_3}\tr{\rho_2\rho_4}$, while $X_3=X_1^*$. 

The generalized suppression laws for Fourier interference of multimode, multiphoton states can now be seen as a corollary of Theorem~\ref{theo:violation_suplaws}. By defining the projector onto the cyclic-symmetric subspace \begin{equation}
    \hat{\Pi}_C=\frac{1}{M}\sum_{k=0}^{M-1}\hat{C}^k, 
\end{equation}
it can be seen that $P_0= \tr{\hat{\Pi}_C \Omega}$, i.e. $P_0$ quantifies the component of the input state that is invariant under cyclic permutations. Hence, we have the following corollary, proved in the Appendix. \begin{cor}\label{cor:sup_laws}
    The input state $\Omega$ of the Fourier interferometry process is invariant under the cyclic permutation, i.e. $\hat{C}\Omega=\Omega$, if and only if $P_0=1$. Consequently, any outcome $\vec{S}$ such that $f(\vec{S}) \neq 0$ is forbidden. 
\end{cor}
The result of \cite{tichy2010zero} is recovered when the systems $\rho_j$ are indistinguishable single photon states. Interestingly, the only assumption to derive Theorem~\ref{theo:violation_suplaws} and Corollary~\ref{cor:sup_laws} is that $\Omega$ is an $M$-party bosonic state belonging to the tensor product of $M$ identical Fock spaces. The state can thus be entangled, which opens 
the way to the study of suppression laws and symmetries of entangled states in future work.

In Fig.~\ref{fig:generalizedHOM} we depict a generalized HOM test ($M=2$) between states $\rho_1$ and $\rho_2$, which could both be described as a superposition of photon numbers in different internal d.o.f., such as spatial, frequency, or polarization. In this scenario, $\rho_1$ and $\rho_2$ interfere at a 50:50 beam splitter and $P_0$ is the probability to measure patterns $\vec{S}$ such that $f(\vec{S})=0$, that is, an even photon number in mode 1. Here, $X_0=1$ and $X_1=\tr{\rho_1\rho_2}$, so that we have $P_0=(1+\tr{\rho_1\rho_2})/2$. Note that Eq.~\eqref{eq:P_bunch_HOM} is recovered when $\rho_1$ and $\rho_2$ are single-photon states, 
since $f(\vec{S})=0$ for bunching events. Thus, our generalized HOM test provides a full optical implementation of the SWAP test compared to \cite{HOM_swap}, as we may also access the overlap between two arbitrary many-photon states (see also \cite{multicopy-cerf-3}). Crucially, our protocol extends beyond pairwise overlaps
and allows us to measure multivariate traces of photonic quantum states, possibly living in an exponentially large Hilbert space. In what follows, we list some selected applications of this protocol.





\begin{figure}[t]
    \centering
    \includegraphics[width=1\linewidth]{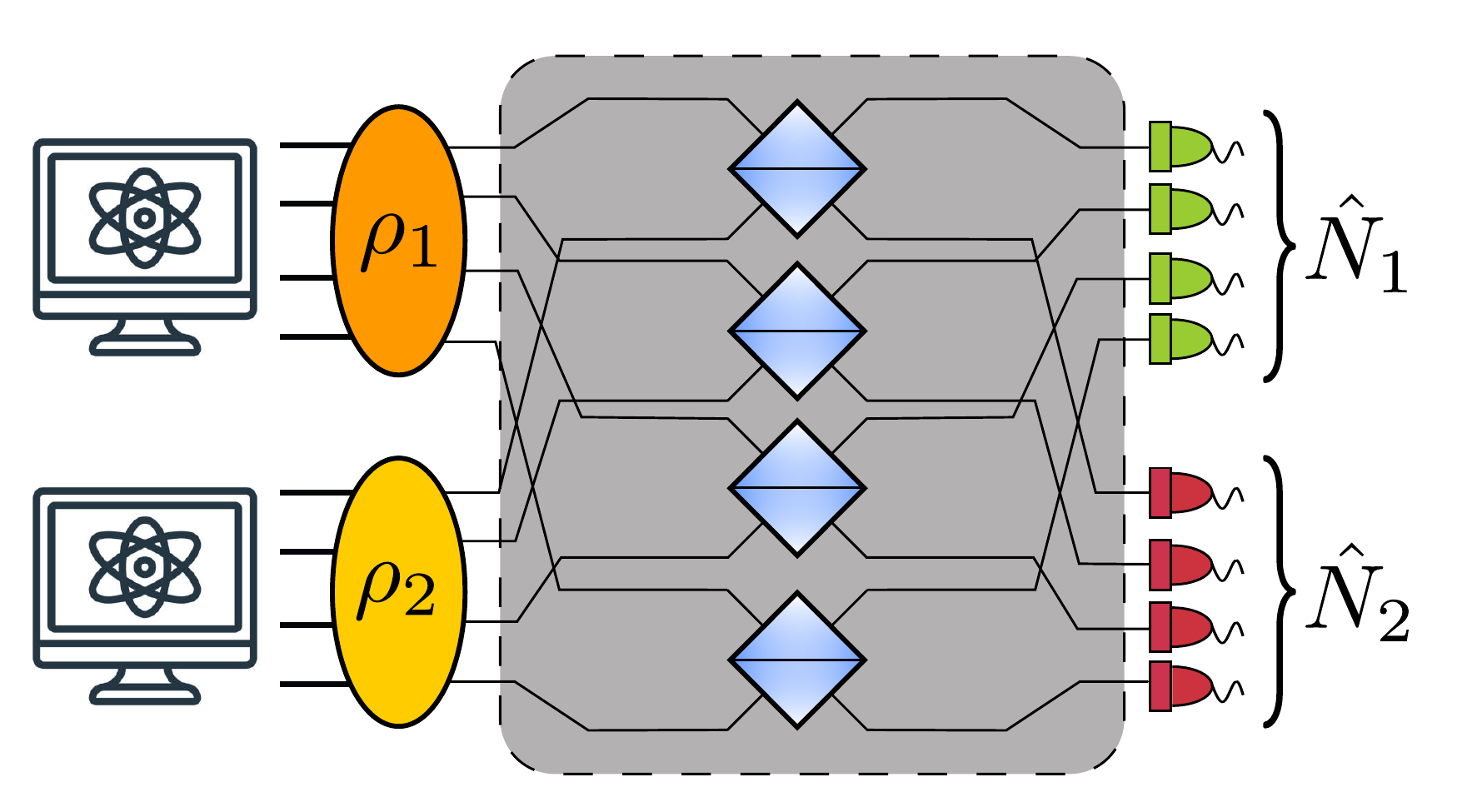}
    \caption{Generalized HOM test for estimation of overlaps of generic input states $\{\rho_1,\rho_2\}$ involving multiple photons occupying several spatial modes, as well as further internal d.o.f. The $i$th spatial mode of $\rho_1$ interferes via a  50:50 beam splitter with the $i$th spatial mode of $\rho_2$. We assume the beam-splitter leaves invariant other internal modes (e.g. frequency, polarization). Total mode occupations $S_j$ are measured by photon counting and outcomes are binned according to $f(\vec{S})$ to estimate probabilities $P_j$ (see Eqs.~\eqref{eq:P_j}), for an estimation of overlap $\tr{\rho_1 \rho_2}$. A generalized HOM suppression is observed if $\rho_1$ and $\rho_2$ are identical pure states for outcomes $\vec{S}$ such that $f(\vec{S})=1$.}
\label{fig:generalizedHOM}
\end{figure}

\paragraph{Quantum kernel methods -- } A common task in machine learning is the supervised learning task, where an algorithm consumes data-label pairs $\{\boldsymbol{x},y\}\in \mathcal{X}\times \{0,1\}$ and outputs a function that ideally classifies following test data. A popular algorithm for this task is the Support Vector Machine (SVM) which is trained on inner products $\langle\boldsymbol{x}_i, \boldsymbol{x}_j \rangle$ in the input space to find a robust linear classification boundary that best separates the data.  Quantum kernel methods (QKMs) show promise for accelerating data analysis by efficiently learning relationships between input data points that have been encoded into an exponentially large Hilbert space~\cite{petersMachineLearningHigh2021}, 
potentially resulting in a kernel function that is expressive and challenging to compute classically \cite{ganUnifiedFrameworkTraceinduced2023,schnabelQuantumKernelMethods2025,schuldQuantumMachineLearning2019,thanasilpExponentialConcentrationQuantum2024}. A usual approach is to define a kernel of the type \cite{anaiContinuousvariableQuantumKernel2024,hochQuantumMachineLearning2025,yinExperimentalQuantumenhancedKernelbased2025,chabaudQuantumMachineLearning2021}
\begin{equation}
    \mathcal{K}(\boldsymbol{x}_i, \boldsymbol{x}_j)=|\langle \phi(\boldsymbol{x}_{i})|\phi(\boldsymbol{x}_{j}) \rangle|^{2} 
\end{equation}
and a classifier function
\begin{equation}    f(\boldsymbol{x})=\operatorname{sign}{\bigg (}\sum_{i}a_{i}y_{i}\mathcal{K}(\boldsymbol{x}_{i},\boldsymbol{x}) + b{\bigg )} \, .
\end{equation}
The usual encoding in quantum computers can be performed by applying a unitary, parametrized by the values $\boldsymbol{x}_{i}$ on a initial state $|\phi(\boldsymbol{x}_{i})\rangle=U(\boldsymbol{x}_{i})|\Psi\rangle$, and as a consequence measuring $\mathcal{K}(\boldsymbol{x}_{i},\boldsymbol{x}_{j})$ reduces to measuring $|\langle \Psi|U^{\dagger}(\boldsymbol{x}_{j})U(\boldsymbol{x}_{i})|\Psi\rangle|^{2}$. Notice that the quantum kernel estimation be done in linear optics by the generalized HOM test, even if the state preparation itself can be done by a more complex procedure, such as adaptive boson sampling~\cite{chabaudQuantumMachineLearning2021} of even by a universal photonic quantum computer. This approach can drastically improve current methods especially for states obtained via postselection or adaptive schemes~\cite{chabaudQuantumMachineLearning2021,hochQuantumMachineLearning2025}.
\\

\paragraph{Eigenvalue estimation --} Using multiple copies of a generic bosonic state $\rho$, i.e. $\Omega=\rho^{\otimes M}$, it is possible to estimate certain spectral properties of the state. In particular by measuring the set $\{\tr{\rho^{j}} \forall j\in [2,n]\}$, it possible to use the Faddeev-LeVerrier algorithm~\cite{baerFaddeevLeVerrierAlgorithmPfaffian2021} to construct the characteristic polynomial, whose roots are the eigenvalues of $\rho$. This task is prone to noise especially for the cases in which $\operatorname{rank}(\rho)$ is large. Nevertheless the task of learning the largest eigenvalue of $\rho$, with high precision, has been shown \cite{wagner2024quantum} to be more resilient to noise and can be, in principle, performed directly from the knowledge of low-order traces. Similarly, it is possible to directly evaluate the Rényi entropy of any state
\begin{equation}
    H_{\alpha}(\rho)=\frac{1}{1-\alpha}\ln \tr{\rho^{\alpha}}
\end{equation}
for $\alpha\in \{n\in\mathbb{N}:n\geq 2\}$, as proposed in \cite{daley2012measuring}.

\paragraph{ Partial photonic distinguishability --} In the study of multi-particle interference, 
particles can often be partially distinguished due to differences in internal degrees of freedom such as polarization or time-of-arrival. This phenomenon strongly affects the outcome transition probabilities of linear interference experiments, like boson sampling \cite{aaronsonComputationalComplexityLinear2010, Tichy2015, shchesnovichPartialIndistinguishabilityTheory2015a}, with quantum computational advantage believed to be possible only in a regime of high indistinguishability \cite{hovenEfficientClassicalAlgorithm2025}. While suppression laws have been proposed as a tool to validate the input state for  boson sampling experiments with single photon input, our results 
show that this idea can be extended to Gaussian boson sampling, which needs indistinguishable single-mode squeezed input states \cite{Hamilton17, shi2022effect}. Moreover, with the new understanding that violations of suppression laws quantify departure from cyclic symmetry of the input state, it is possible to obtain precise bounds on the closeness of the input state to a state of identical particles, as well as quantify genuine multiphoton indistinguishability \cite{Brod_19, englbrechtIndistinguishabilityIdenticalBosons2024}. These results will be presented in a separate work \cite{Robbio2026}.  Very recently, an independent work by Sanz et al. \cite{sanz2026exponential} also proposes the use of Fourier interferometers for efficient benchmarking of genuine multiphoton indistinguishability. We stress that current alternative techniques for quantifying genuine multiphoton indistinguishability are inefficient, as they require postselection on an outcome with exponentially small probability \cite{Pont_22}.  

It is known that the outcomes of multiphoton interference experiments with single photon inputs depend only on multivariate traces of the input states 
\cite{shchesnovichPartialIndistinguishabilityTheory2015a, shchesnovich2018collective, Oszmaniec2024}. Our tools provide a systematic way to fully reconstruct this relational information. In fact, for pure state single-photon inputs, with internal d.o.f. described by wavefunctions $|\phi_j\rangle$, this relational information is fully described by a Gram matrix $\langle \phi_j|\phi_k\rangle$. It has been shown that the all physically relevant elements of the Gram matrix, can be characterized via the state overlaps $|\mathcal{S}_{jk}|^2= |\langle \phi_j|\phi_k\rangle|^2$, measurable from HOM tests, as well as at most $O(n^2)$ collective photonic phases \cite{Oszmaniec2024}. These collective photonic phases are, in fact, the phases~\cite{Menssen_17, shchesnovich2018collective} of complex-valued multivariate traces that affect outcome probabilities. The efficient protocol we provide for measuring multivariate traces 
can be applied at most $O(M^2 )$ times in order to reconstruct all physically relevant collective photonic phases in arbitrary partial distinguishability scenarios featuring pure states \cite{Oszmaniec2024}. In addition, collective photonic phases are of interest due to their relation to the more general concept of geometric phases \cite{pancharatnam1956generalized, Berry09}, and have been shown to lead to counterintuitive phenomena in bosonic interference \cite{Menssen_17, Menssen_22, rodari_24_counter}, playing a role in practical protocols such as indistinguishability distillation \cite{hoch2025optimal}.


\paragraph{Discussion --} 

Multiple other applications of our scheme for measuring multivariate traces may be anticipated, as physically relevant quantities can often be represented as simple functions of low-order Bargmann invariants. For example, in the Appendix we discuss how our protocol gives direct measurement schemes for the Wigner function, the Husimi-Q distribution,  Kirkwood-Dirac quasi-probabilities \cite{arvidsson-shukurPropertiesApplicationsKirkwoodDirac2024}, and the P-representation, all representable by invariants of at most third-order. Weak values \cite{wagner2023simple}, useful in quantum metrology, are also functions of 3rd-order invariants, and  out-of-time-order correlators (OTOCs) \cite{yunger18}, useful for characterizing information scrambling, are given by 5th-order invariants \cite{wagner2024quantum}. 

In conclusion, we have presented a native linear-optical protocol that uses  Fourier interferometry, photon-counting and simple classical postprocessing to directly estimates multivariate traces (Bargmann invariants) of general bosonic input states. We have also generalized suppression laws in Fourier interferometry to multiphoton, multimode states and shown how violations of suppression laws quantify departure from cyclic symmetry and can be predicted via the relational information of these input states. Our results provide a strong equivalence between the computational power of the cycle test in the circuit model and that of Fourier interferometry in linear-optics, showing that a useful primitive in quantum computing can be done in photonic system purely with linear optics, avoiding costly photonic entangling gates. This provides a significant generalization of the well-known equivalence between SWAP and HOM tests, which had been discussed mostly for single-photon states, opening up several applications which can be explored in near-term photonic quantum computing.

\paragraph{Acknowledgments --}
N.J.C. thanks M. Arnhem, S. De Bièvre and C. Griffet for insightful discussions that inspired this work. L.N. and E.F.G. acknowledge support from FCT-Fundação para a Ciência e a Tecnologia (Portugal) via Project No. CEECINST/00062/2018. L.N. acknowledges funding from the European Union’s Horizon Europe Framework Programme (EIC Pathfinder Challenge project Veriqub) under Grant Agreement No. 101114899. M.R. is a FRIA grantee of the Fonds de la Recherche Scientifique – FNRS. E.F.G.  acknowledges funding from the National Council for Scientific and Technological Development – CNPq (Brazil) under grant 308292/2025-1. N.J.C. acknowledges support from the Fonds de la Recherche Scientifique – FNRS under project CHEQS within the Excellence of Science (EOS) program.

\bibliography{biblio}

@article{lutterbach97,
  title = {Method for Direct Measurement of the Wigner Function in Cavity QED and Ion Traps},
  author = {Lutterbach, L. G. and Davidovich, L.},
  journal = {Phys. Rev. Lett.},
  volume = {78},
  issue = {13},
  pages = {2547--2550},
  numpages = {0},
  year = {1997},
  month = {Mar},
  publisher = {American Physical Society},
  doi = {10.1103/PhysRevLett.78.2547},
  url = {https://link.aps.org/doi/10.1103/PhysRevLett.78.2547}
}

@article{Hamilton17,
  title = {Gaussian Boson Sampling},
  author = {Hamilton, Craig S. and Kruse, Regina and Sansoni, Linda and Barkhofen, Sonja and Silberhorn, Christine and Jex, Igor},
  journal = {Phys. Rev. Lett.},
  volume = {119},
  issue = {17},
  pages = {170501},
  numpages = {5},
  year = {2017},
  month = {Oct},
  publisher = {American Physical Society},
  doi = {10.1103/PhysRevLett.119.170501},
  url = {https://link.aps.org/doi/10.1103/PhysRevLett.119.170501}
}

@article{Robbio2026,
  title = {in preparation},
journal = {in preparation},
  author = {Robbio, M and Oszmaniec, M and Galvão, E F and Cerf, N J and Novo, L},
}

@article{sanz2026exponential,
  title={Exponential improvement in benchmarking multiphoton interference},
  author={Sanz, Rodrigo M and Annoni, Emilio and Wein, Stephen C and Almudever, Carmen G and Mansfield, Shane and Derbyshire, Ellen and Mezher, Rawad},
  journal={arXiv preprint arXiv:2601.10289},
  year={2026},
    doi={https://doi.org/10.48550/arXiv.2601.10289}
}

@article{salavrakos2025photon_native,
  title={Photon-native quantum algorithms},
  author={Salavrakos, Alexia and Maring, Nicolas and Emeriau, Pierre-Emmanuel and Mansfield, Shane},
  journal={Materials for Quantum Technology},
  volume={5},
  number={2},
  pages={023001},
  year={2025},
  publisher={IOP Publishing},
doi={10.1088/2633-4356/adc531}
}

@article{yunger18,
  title={Quasiprobability behind the out-of-time-ordered correlator},
  author={Yunger Halpern, Nicole and Swingle, Brian and Dressel, Justin},
  journal={Physical Review A},
  volume={97},
  number={4},
  pages={042105},
  year={2018},
  publisher={APS},
doi={https://doi.org/10.1103/PhysRevA.97.042105}
}

@article{arvidsson24,
  title={Properties and applications of the Kirkwood--Dirac distribution},
  author={Arvidsson-Shukur, David RM and Braasch Jr, William F and De Bievre, Stephan and Dressel, Justin and Jordan, Andrew N and Langrenez, Christopher and Lostaglio, Matteo and Lundeen, Jeff S and Halpern, Nicole Yunger},
  journal={New Journal of Physics},
  volume={26},
  number={12},
  pages={121201},
  year={2024},
  publisher={IOP Publishing},
doi={https://doi.org/10.1088/1367-2630/ada05d}
}

@article{viggianiello18,
  title={Experimental generalized quantum suppression law in Sylvester interferometers},
  author={Viggianiello, Niko and Flamini, Fulvio and Innocenti, Luca and Cozzolino, Daniele and Bentivegna, Marco and Spagnolo, Nicol{\`o} and Crespi, Andrea and Brod, Daniel J and Galv{\~a}o, Ernesto F and Osellame, Roberto and others},
  journal={New Journal of Physics},
  volume={20},
  number={3},
  pages={033017},
  year={2018},
  publisher={IOP Publishing},
doi={https://doi.org/10.1088/1367-2630/aaad92}
}

@article{crespi15,
  title={Suppression laws for multiparticle interference in Sylvester interferometers},
  author={Crespi, Andrea},
  journal={Physical Review A},
  volume={91},
  number={1},
  pages={013811},
  year={2015},
  publisher={APS},
doi={https://doi.org/10.1103/PhysRevA.91.013811}
}

@article{tichy2010zero,
  title={Zero-transmission law for multiport beam splitters},
  author={Tichy, Malte Christopher and Tiersch, Markus and de Melo, Fernando and Mintert, Florian and Buchleitner, Andreas},
  journal={Physical review letters},
  volume={104},
  number={22},
  pages={220405},
  year={2010},
  publisher={APS},
doi={https://doi.org/10.1103/PhysRevLett.104.220405}
}

@misc{annoniIncoherentBehaviorPartially2025,
  title = {Incoherent Behavior of Partially Distinguishable Photons},
  author = {Annoni, Emilio and Wein, Stephen C.},
  year = 2025,
  month = sep,
  number = {arXiv:2502.05047},
  eprint = {2502.05047},
  primaryclass = {quant-ph},
  publisher = {arXiv},
  doi = {10.48550/arXiv.2502.05047},
  urldate = {2025-10-07},
  archiveprefix = {arXiv}
}

@misc{clementsOptimalDesignUniversal2017,
  title = {An {{Optimal Design}} for {{Universal Multiport Interferometers}}},
  author = {Clements, William R. and Humphreys, Peter C. and Metcalf, Benjamin J. and Kolthammer, W. Steven and Walmsley, Ian A.},
  year = 2017,
  month = feb,
  number = {arXiv:1603.08788},
  eprint = {1603.08788},
  primaryclass = {physics},
  publisher = {arXiv},
  doi = {10.48550/arXiv.1603.08788},
  urldate = {2025-11-27},
  archiveprefix = {arXiv}
}

@article{reckExperimentalRealizationAny1994,
  title = {Experimental Realization of Any Discrete Unitary Operator},
  author = {Reck, Michael and Zeilinger, Anton and Bernstein, Herbert J. and Bertani, Philip},
  year = 1994,
  month = jul,
  journal = {Physical Review Letters},
  volume = {73},
  number = {1},
  pages = {58--61},
  publisher = {American Physical Society},
  doi = {10.1103/PhysRevLett.73.58},
  urldate = {2025-11-27}
}

@book{carmichael2008statistical,
  title={Statistical methods in quantum optics 2: Non-classical fields},
  author={Carmichael, Howard J},
  year={2008},
  publisher={Springer},
doi={https://doi.org/10.1007/978-3-540-71320-3}
}

@article{diracAnalogyClassicalQuantum1945,
  title = {On the {{Analogy Between Classical}} and {{Quantum Mechanics}}},
  author = {Dirac, P. A. M.},
  year = 1945,
  month = apr,
  journal = {Reviews of Modern Physics},
  volume = {17},
  number = {2-3},
  pages = {195--199},
  publisher = {American Physical Society},
  doi = {10.1103/RevModPhys.17.195},
  urldate = {2025-12-09}
}

@article{kirkwoodQuantumStatisticsAlmost1933,
  title = {Quantum {{Statistics}} of {{Almost Classical Assemblies}}},
  author = {Kirkwood, John G.},
  year = 1933,
  month = jul,
  journal = {Physical Review},
  volume = {44},
  number = {1},
  pages = {31--37},
  publisher = {American Physical Society},
  doi = {10.1103/PhysRev.44.31},
  urldate = {2025-12-09},
  doi={https://doi.org/10.1103/PhysRev.44.31}
}

@article{shchesnovichPartialIndistinguishabilityTheory2015a,
  title = {Partial Indistinguishability Theory for Multi-Photon Experiments in Multiport Devices},
  author = {Shchesnovich, V. S.},
  year = 2015,
  month = jan,
  journal = {Physical Review A},
  volume = {91},
  number = {1},
  eprint = {1410.1506},
  primaryclass = {quant-ph},
  pages = {013844},
  issn = {1050-2947, 1094-1622},
  doi = {10.1103/PhysRevA.91.013844},
  urldate = {2025-11-04},
  archiveprefix = {arXiv}
}

@article{englbrechtIndistinguishabilityIdenticalBosons2024,
  title = {Indistinguishability of {{Identical Bosons}} from a {{Quantum Information Theory Perspective}}},
  author = {Englbrecht, Matthias and Kraft, Tristan and Dittel, Christoph and Buchleitner, Andreas and Giedke, Geza and Kraus, Barbara},
  year = {2024},
  month = jan,
  journal = {Physical Review Letters},
  volume = {132},
  number = {5},
  pages = {050201},
  publisher = {American Physical Society},
  doi = {10.1103/PhysRevLett.132.050201},
  urldate = {2025-09-11}
}

@article{hilleryDISTRIBUTIONFUNCTIONSPHYSICS,
   title={Distribution functions in physics: Fundamentals},
  author={Hillery, MOSM and O'Connell, Robert F and Scully, Marlan O and Wigner, Eugene P},
  journal={Physics reports},
  volume={106},
  number={3},
  pages={121--167},
  year={1984},
  publisher={Elsevier},
doi={https://doi.org/10.1016/0370-1573(84)90160-1}
}

@article{hongMeasurementSubpicosecondTime1987,
  title = {Measurement of Subpicosecond Time Intervals between Two Photons by Interference},
  author = {Hong, C. K. and Ou, Z. Y. and Mandel, L.},
  year = {1987},
  month = nov,
  journal = {Physical Review Letters},
  volume = {59},
  number = {18},
  pages = {2044--2046},
  publisher = {American Physical Society},
  doi = {10.1103/PhysRevLett.59.2044},
  urldate = {2025-09-11}
}

@article{arvidsson-shukurPropertiesApplicationsKirkwoodDirac2024,
  title = {Properties and {{Applications}} of the {{Kirkwood-Dirac Distribution}}},
  author = {{Arvidsson-Shukur}, David R. M. and Braasch, William F. and Bievre, Stephan De and Dressel, Justin and Jordan, Andrew N. and Langrenez, Christopher and Lostaglio, Matteo and Lundeen, Jeff S. and Halpern, Nicole Yunger},
  year = 2024,
  month = dec,
  journal = {New Journal of Physics},
  volume = {26},
  number = {12},
  eprint = {2403.18899},
  primaryclass = {quant-ph},
  pages = {121201},
  issn = {1367-2630},
  doi = {10.1088/1367-2630/ada05d},
  urldate = {2025-12-10},
  archiveprefix = {arXiv},
  langid = {english}
}

@incollection{curilefHusimiDistributionDevelopment2013,
  title = {The {{Husimi Distribution}}: {{Development}} and {{Applications}}},
  shorttitle = {The {{Husimi Distribution}}},
  booktitle = {Advances in {{Quantum Mechanics}}},
  author = {Curilef, Sergio and Pennini, Flavia},
  editor = {Bracken, Paul},
  year = 2013,
  month = apr,
  publisher = {InTech},
  doi = {10.5772/53846},
  urldate = {2025-12-10},
  isbn = {978-953-51-1089-7},
  langid = {english}
}

@article{ferrieQuasiprobabilityRepresentationsQuantum2011,
  title = {Quasi-Probability Representations of Quantum Theory with Applications to Quantum Information Science},
  author = {Ferrie, Christopher},
  year = 2011,
  month = nov,
  journal = {Reports on Progress in Physics},
  volume = {74},
  number = {11},
  eprint = {1010.2701},
  primaryclass = {quant-ph},
  pages = {116001},
  issn = {0034-4885, 1361-6633},
  doi = {10.1088/0034-4885/74/11/116001},
  urldate = {2025-12-10},
  archiveprefix = {arXiv}
}

@article{Berry09,
author = {M. V. Berry and S. Klein},
title = {Geometric phases from stacks of crystal plates},
journal = {Journal of Modern Optics},
volume = {43},
number = {1},
pages = {165-180},
year = {1996},
publisher = {Taylor & Francis},
doi = {10.1080/09500349608232731},
URL ={https://www.tandfonline.com/doi/abs/10.1080/09500349608232731}
}

@article{cahillOrderedExpansionsBoson1969,
  title = {Ordered {{Expansions}} in {{Boson Amplitude Operators}}},
  author = {Cahill, K. E. and Glauber, R. J.},
  year = 1969,
  month = jan,
  journal = {Physical Review},
  volume = {177},
  number = {5},
  pages = {1857--1881},
  issn = {0031-899X},
  doi = {10.1103/PhysRev.177.1857},
  urldate = {2025-12-09},
  copyright = {http://link.aps.org/licenses/aps-default-license},
  langid = {english}
}

@article{quek2024multivariate,
  title={Multivariate trace estimation in constant quantum depth},
  author={Quek, Yihui and Kaur, Eneet and Wilde, Mark M},
  journal={Quantum},
  volume={8},
  pages={1220},
  year={2024},
  publisher={Verein zur F{\"o}rderung des Open Access Publizierens in den Quantenwissenschaften},
doi={https://doi.org/10.48550/arXiv.2206.15405}
}

@article{shchesnovich2018collective,
  title={Collective phases of identical particles interfering on linear multiports},
  author={Shchesnovich, VS and Bezerra, MEO},
  journal={Physical Review A},
  volume={98},
  number={3},
  pages={033805},
  year={2018},
  publisher={APS},
doi={https://doi.org/10.1103/PhysRevA.98.033805}
}

@article{dittel2018totally,
  title={Totally destructive many-particle interference},
  author={Dittel, Christoph and Dufour, Gabriel and Walschaers, Mattia and Weihs, Gregor and Buchleitner, Andreas and Keil, Robert},
  journal={Physical Review Letters},
  volume={120},
  number={24},
  pages={240404},
  year={2018},
  publisher={APS},
doi={https://doi.org/10.1103/PhysRevA.97.062116}
}

@article{daley2012measuring,
  title={Measuring entanglement growth in quench dynamics of bosons in an optical lattice},
  author={Daley, Andrew J and Pichler, Hannes and Schachenmayer, Johannes and Zoller, Peter},
  journal={Physical review letters},
  volume={109},
  number={2},
  pages={020505},
  year={2012},
  publisher={APS},
doi={https://doi.org/10.1103/PhysRevLett.109.020505}
}

@article{wu2022sparse,
  title={Sparse interferometry for measuring multiphoton collective phase},
  author={Wu, Jizhou and Sanders, Barry C},
  journal={Physical Review Research},
  volume={4},
  number={2},
  pages={023134},
  year={2022},
  publisher={APS},
doi={https://doi.org/10.1103/PhysRevResearch.4.023134}
}

@article{ekert2002direct,
  title={Direct estimations of linear and nonlinear functionals of a quantum state},
  author={Ekert, Artur K and Alves, Carolina Moura and Oi, Daniel KL and Horodecki, Micha{\l} and Horodecki, Pawe{\l} and Kwek, Leong Chuan},
  journal={Physical Review Letters},
  volume={88},
  number={21},
  pages={217901},
  year={2002},
  month = {May},
  publisher = {American Physical Society},
  doi = {10.1103/PhysRevLett.88.217901},
  url = {https://link.aps.org/doi/10.1103/PhysRevLett.88.217901}
}

@article{ekert2003direct,
  title = {Direct estimation of functionals of density operators by local operations and classical communication},
  author = {Alves, Carolina Moura and Horodecki, Pawe\l{} and Oi, Daniel K. L. and Kwek, L. C. and Ekert, Artur K.},
  journal = {Phys. Rev. A},
  volume = {68},
  issue = {3},
  pages = {032306},
  numpages = {4},
  year = {2003},
  month = {Sep},
  publisher = {American Physical Society},
  doi = {10.1103/PhysRevA.68.032306},
  url = {https://link.aps.org/doi/10.1103/PhysRevA.68.032306}
}

@article{brun2004,
  title = {Measuring polynomial functions of states},
  author = {Brun, Todd A.},
  journal = {Quantum Inf. Comput.},
  volume = {4},
  pages = {401},
  year = {2004},
doi={https://doi.org/10.48550/arXiv.quant-ph/0401067}
}

@article{baerFaddeevLeVerrierAlgorithmPfaffian2021,
  title = {The {{Faddeev-LeVerrier}} Algorithm and the {{Pfaffian}}},
  author = {Baer, Christian},
  year = 2021,
  month = dec,
  journal = {Linear Algebra and its Applications},
  volume = {630},
  eprint = {2008.04247},
  primaryclass = {math},
  pages = {39--55},
  issn = {00243795},
  doi = {10.1016/j.laa.2021.07.023},
  urldate = {2025-12-17},
  archiveprefix = {arXiv}
}

@article{Oszmaniec2024,
  title = {Measuring relational information between quantum states,  and applications},
  volume = {26},
  ISSN = {1367-2630},
  url = {http://dx.doi.org/10.1088/1367-2630/ad1a27},
  DOI = {10.1088/1367-2630/ad1a27},
  number = {1},
  journal = {New Journal of Physics},
  publisher = {IOP Publishing},
  author = {Oszmaniec,  Michał and Brod,  Daniel J and Galvão,  Ernesto F},
  year = {2024},
  month = jan,
  pages = {013053}
}

@article{HOM_swap,
  title = {swap test and Hong-Ou-Mandel effect are equivalent},
  author = {Garcia-Escartin, Juan Carlos and Chamorro-Posada, Pedro},
  journal = {Physical Review A},
  volume = {87},
  issue = {5},
  pages = {052330},
  numpages = {10},
  year = {2013},
  month = {May},
  publisher = {American Physical Society},
  doi = {10.1103/PhysRevA.87.052330},
  url = {https://link.aps.org/doi/10.1103/PhysRevA.87.052330}
}

@article{shi2022effect,
  title={Effect of partial distinguishability on quantum supremacy in Gaussian Boson sampling},
  author={Shi, Junheng and Byrnes, Tim},
  journal={npj Quantum Information},
  volume={8},
  number={1},
  pages={54},
  year={2022},
  publisher={Nature Publishing Group UK London},
doi={https://doi.org/10.1038/s41534-022-00557-9}
}

@article{Knill2001,
  title = {A scheme for efficient quantum computation with linear optics},
  volume = {409},
  ISSN = {1476-4687},
  url = {http://dx.doi.org/10.1038/35051009},
  DOI = {10.1038/35051009},
  number = {6816},
  journal = {Nature},
  publisher = {Springer Science and Business Media LLC},
  author = {Knill,  E. and Laflamme,  R. and Milburn,  G. J.},
  year = {2001},
  month = jan,
  pages = {46–52}
}

@article{Bartolucci2023,
  title = {Fusion-based quantum computation},
  volume = {14},
  pages = {912},
  ISSN = {2041-1723},
  url = {http://dx.doi.org/10.1038/s41467-023-36493-1},
  DOI = {10.1038/s41467-023-36493-1},
  number = {1},
  journal = {Nature Communications},
  publisher = {Springer Science and Business Media LLC},
  author = {Bartolucci,  Sara and Birchall,  Patrick and Bombín,  Hector and Cable,  Hugo and Dawson,  Chris and Gimeno-Segovia,  Mercedes and Johnston,  Eric and Kieling,  Konrad and Nickerson,  Naomi and Pant,  Mihir and Pastawski,  Fernando and Rudolph,  Terry and Sparrow,  Chris},
  year = {2023},
  month = feb 
}

@article{lindner2009proposal,
  title={Proposal for pulsed on-demand sources of photonic cluster state strings},
  author={Lindner, Netanel H and Rudolph, Terry},
  journal={Physical review letters},
  volume={103},
  number={11},
  pages={113602},
  year={2009},
  publisher={APS},
doi={https://doi.org/10.1103/PhysRevLett.103.113602}
}

@article{huet2025deterministic,
  title={Deterministic and reconfigurable graph state generation with a single solid-state quantum emitter},
  author={Huet, H and Ramesh, PR and Wein, SC and Coste, N and Hilaire, P and Somaschi, N and Morassi, M and Lema{\^\i}tre, A and Doty, MF and others},
  journal={Nature communications},
  volume={16},
  number={1},
  pages={4337},
  year={2025},
  publisher={Nature Publishing Group UK London},
doi={https://doi.org/10.1038/s41467-025-61093-6}
}

@article{Pont_22,
  title = {Quantifying $n$-Photon Indistinguishability with a Cyclic Integrated Interferometer},
  author = {Pont, Mathias and Albiero, Riccardo and Thomas, Sarah E. and Spagnolo, Nicol\`o and Ceccarelli, Francesco and Corrielli, Giacomo and Brieussel, Alexandre and Somaschi, Niccolo and Huet, H\^elio and Harouri, Abdelmounaim and Lema\^{\i}tre, Aristide and Sagnes, Isabelle and Belabas, Nadia and Sciarrino, Fabio and Osellame, Roberto and Senellart, Pascale and Crespi, Andrea},
  journal = {Physical Review X},
  volume = {12},
  issue = {3},
  pages = {031033},
  numpages = {22},
  year = {2022},
  month = {Sep},
  publisher = {American Physical Society},
  doi = {10.1103/PhysRevX.12.031033},
  url = {https://link.aps.org/doi/10.1103/PhysRevX.12.031033}
}

@article{Menssen_17,
  title = {Distinguishability and Many-Particle Interference},
  author = {Menssen, Adrian J. and Jones, Alex E. and Metcalf, Benjamin J. and Tichy, Malte C. and Barz, Stefanie and Kolthammer, W. Steven and Walmsley, Ian A.},
  journal = {Physical Review Letters},
  volume = {118},
  issue = {15},
  pages = {153603},
  numpages = {6},
  year = {2017},
  month = {Apr},
  publisher = {American Physical Society},
  doi = {10.1103/PhysRevLett.118.153603},
  url = {https://link.aps.org/doi/10.1103/PhysRevLett.118.153603}
}

@article{Brod_19,
  title = {Witnessing Genuine Multiphoton Indistinguishability},
  author = {Brod, Daniel J. and Galv\~ao, Ernesto F. and Viggianiello, Niko and Flamini, Fulvio and Spagnolo, Nicol\`o and Sciarrino, Fabio},
  journal = {Physical Review Letters},
  volume = {122},
  issue = {6},
  pages = {063602},
  numpages = {7},
  year = {2019},
  month = {Feb},
  publisher = {American Physical Society},
  doi = {10.1103/PhysRevLett.122.063602},
  url = {https://link.aps.org/doi/10.1103/PhysRevLett.122.063602}
}

@inproceedings{pancharatnam1956generalized,
  title={Generalized theory of interference, and its applications: Part I. Coherent pencils},
  author={Pancharatnam, Shivaramakrishnan},
  booktitle={Proceedings of the Indian Academy of Sciences-Section A},
  volume={44},
  number={5},
  pages={247--262},
  year={1956},
  organization={Springer},
doi={https://doi.org/10.1007/BF03046050}
}

@article{Menssen_22,
  title = {Multiparticle Interference of Pairwise Distinguishable Photons},
  author = {Jones, Alex E. and Menssen, Adrian J. and Chrzanowski, Helen M. and Wolterink, Tom A. W. and Shchesnovich, Valery S. and Walmsley, Ian A.},
  journal = {Physical Review Letters},
  volume = {125},
  issue = {12},
  pages = {123603},
  numpages = {6},
  year = {2020},
  month = {Sep},
  publisher = {American Physical Society},
  doi = {10.1103/PhysRevLett.125.123603},
  url = {https://link.aps.org/doi/10.1103/PhysRevLett.125.123603}
}

@misc{rodari_24_counter,
      title={Experimental observation of counter-intuitive features of photonic bunching}, 
      author={Giovanni Rodari and Carlos Fernandes and Eugenio Caruccio and Alessia Suprano and Francesco Hoch and Taira Giordani and Gonzalo Carvacho and Riccardo Albiero and Niki Di Giano and Giacomo Corrielli and Francesco Ceccarelli and Roberto Osellame and Daniel J. Brod and Leonardo Novo and Nicolò Spagnolo and Ernesto F. Galvão and Fabio Sciarrino},
      year={2024},
      eprint={2410.15883},
      archivePrefix={arXiv},
      primaryClass={quant-ph},
      url={https://arxiv.org/abs/2410.15883}, 
}

@article{wagner2023simple,
  title={Simple proof that anomalous weak values require coherence},
  author={Wagner, Rafael and Galv{\~a}o, Ernesto F},
  journal={Physical Review A},
  volume={108},
  number={4},
  pages={L040202},
  year={2023},
  publisher={APS},
doi={https://doi.org/10.1103/PhysRevA.108.L040202}
}

@article{Fernandes_Barg,
  title = {Unitary-Invariant Witnesses of Quantum Imaginarity},
  author = {Fernandes, Carlos and Wagner, Rafael and Novo, Leonardo and Galv\~ao, Ernesto F.},
  journal = {Physical Review Letters},
  volume = {133},
  issue = {19},
  pages = {190201},
  numpages = {7},
  year = {2024},
  month = {Nov},
  publisher = {American Physical Society},
  doi = {10.1103/PhysRevLett.133.190201},
  url = {https://link.aps.org/doi/10.1103/PhysRevLett.133.190201}
}

@article{hoch2025optimal,
  title={Optimal distillation of photonic indistinguishability},
  author={Hoch, Francesco and Camillini, Anita and Rodari, Giovanni and Caruccio, Eugenio and Carvacho, Gonzalo and Giordani, Taira and Albiero, Riccardo and Di Giano, Niki and Corrielli, Giacomo and Ceccarelli, Francesco and others},
  journal={arXiv preprint arXiv:2509.02296},
  year={2025},
    doi={https://arxiv.org/abs/2509.02296v1}
}

@article{kirkwood1933quantum,
  title={Quantum statistics of almost classical assemblies},
  author={Kirkwood, John G},
  journal={Physical Review},
  volume={44},
  number={1},
  pages={31},
  year={1933},
  publisher={APS},
doi={https://doi.org/10.1103/PhysRev.44.31}
}

@incollection{chruscinski2004geometric,
  title = {Adiabatic {{Phases}} in {{Classical Mechanics}}},
  booktitle = {Geometric {{Phases}} in {{Classical}} and {{Quantum Mechanics}}},
  author = {Chru{\'s}ci{\'n}ski, Dariusz and Jamio{\l}kowski, Andrzej},
  editor = {Chru{\'s}ci{\'n}ski, Dariusz and Jamio{\l}kowski, Andrzej},
  year = {2004},
  pages = {111--155},
  publisher = {Birkh{\"a}user},
  address = {Boston, MA},
  doi = {10.1007/978-0-8176-8176-0_3},
  urldate = {2025-09-23},
  isbn = {978-0-8176-8176-0},
  langid = {english}
}

@article{bamber2014observing,
  title={Observing Dirac’s classical phase space analog to the quantum state},
  author={Bamber, Charles and Lundeen, Jeff S},
  journal={Physical review letters},
  volume={112},
  number={7},
  pages={070405},
  year={2014},
  publisher={APS},
  doi={ https://doi.org/10.1103/PhysRevLett.112.070405}
}

@article{simon1993bargmann,
  title={Bargmann invariant and the geometry of the G{\"u}oy effect},
  author={Simon, R and Mukunda, N},
  journal={Physical review letters},
  volume={70},
  number={7},
  pages={880},
  year={1993},
  publisher={APS},
doi={https://doi.org/10.1103/PhysRevLett.70.880}
}

@article{wagner2024quantum,
  title={Quantum circuits for measuring weak values, Kirkwood--Dirac quasiprobability distributions, and state spectra},
  author={Wagner, Rafael and Schwartzman-Nowik, Zohar and Paiva, Ismael L and Te’eni, Amit and Ruiz-Molero, Antonio and Barbosa, Rui Soares and Cohen, Eliahu and Galv{\~a}o, Ernesto F},
  journal={Quantum Science and Technology},
  volume={9},
  number={1},
  pages={015030},
  year={2024},
  publisher={IOP Publishing},
doi={https://doi.org/10.48550/arXiv.2302.00705}
}

@article{Tichy2015,
  title = {Sampling of partially distinguishable bosons and the relation to the multidimensional permanent},
  volume = {91},
  pages = {022316},
  ISSN = {1094-1622},
  url = {http://dx.doi.org/10.1103/PhysRevA.91.022316},
  DOI = {10.1103/physreva.91.022316},
  number = {2},
  journal = {Physical Review A},
  publisher = {American Physical Society (APS)},
  author = {Tichy,  Malte C.},
  year = {2015},
  month = feb 
}

@article{liang2023unified,
  title={Unified multivariate trace estimation and quantum error mitigation},
  author={Liang, Jin-Min and Lv, Qiao-Qiao and Wang, Zhi-Xi and Fei, Shao-Ming},
  journal={Physical Review A},
  volume={107},
  number={1},
  pages={012606},
  year={2023},
  publisher={APS},
doi={https://doi.org/10.1103/PhysRevA.107.012606}
}

@misc{ganUnifiedFrameworkTraceinduced2023,
  title = {A {{Unified Framework}} for {{Trace-induced Quantum Kernels}}},
  author = {Gan, Beng Yee and Leykam, Daniel and Thanasilp, Supanut},
  year = 2023,
  month = nov,
  number = {arXiv:2311.13552},
  eprint = {2311.13552},
  primaryclass = {quant-ph},
  publisher = {arXiv},
  doi = {10.48550/arXiv.2311.13552},
  urldate = {2025-12-09},
  archiveprefix = {arXiv}
}

@misc{hovenEfficientClassicalAlgorithm2025,
  title = {Efficient Classical Algorithm for Simulating Boson Sampling with Heterogeneous Partial Distinguishability},
  author = {van den Hoven, S. N. and Kanis, E. and Renema, J. J.},
  year = 2025,
  month = mar,
  number = {arXiv:2406.17682},
  eprint = {2406.17682},
  primaryclass = {quant-ph},
  publisher = {arXiv},
  doi = {10.48550/arXiv.2406.17682},
  urldate = {2025-09-11},
  archiveprefix = {arXiv}
}

@article{tichyInterferenceIdenticalParticles2014,
  title = {Interference of Identical Particles from Entanglement to Boson-Sampling},
  author = {Tichy, Malte C},
  year = 2014,
  month = may,
  journal = {Journal of Physics B: Atomic, Molecular and Optical Physics},
  volume = {47},
  number = {10},
  pages = {103001},
  issn = {0953-4075, 1361-6455},
  doi = {10.1088/0953-4075/47/10/103001},
  urldate = {2025-09-11},
  copyright = {http://iopscience.iop.org/info/page/text-and-data-mining},
  langid = {english}
}

@misc{aaronsonComputationalComplexityLinear2010,
  title = {The {{Computational Complexity}} of {{Linear Optics}}},
  author = {Aaronson, Scott and Arkhipov, Alex},
  year = 2010,
  month = nov,
  number = {arXiv:1011.3245},
  eprint = {1011.3245},
  primaryclass = {quant-ph},
  publisher = {arXiv},
  doi = {10.48550/arXiv.1011.3245},
  urldate = {2025-09-11},
  archiveprefix = {arXiv}
}

@article{anaiContinuousvariableQuantumKernel2024,
  title = {Continuous-Variable Quantum Kernel Method on a Programmable Photonic Quantum Processor},
  author = {Anai, Keitaro and Ikehara, Shion and Yano, Yoshichika and Okuno, Daichi and Takeda, Shuntaro},
  year = 2024,
  month = aug,
  journal = {Physical Review A},
  volume = {110},
  number = {2},
  eprint = {2405.01086},
  primaryclass = {quant-ph},
  pages = {022404},
  issn = {2469-9926, 2469-9934},
  doi = {10.1103/PhysRevA.110.022404},
  urldate = {2025-12-15},
  archiveprefix = {arXiv}
}

@article{chabaudQuantumMachineLearning2021,
  title = {Quantum Machine Learning with Adaptive Linear Optics},
  author = {Chabaud, Ulysse and Markham, Damian and Sohbi, Adel},
  year = 2021,
  month = jul,
  journal = {Quantum},
  volume = {5},
  eprint = {2102.04579},
  primaryclass = {quant-ph},
  pages = {496},
  issn = {2521-327X},
  doi = {10.22331/q-2021-07-05-496},
  urldate = {2025-12-15},
  archiveprefix = {arXiv},
  langid = {english}
}

@article{hochQuantumMachineLearning2025,
  title = {Quantum Machine Learning with {{Adaptive Boson Sampling}} via Post-Selection},
  author = {Hoch, Francesco and Caruccio, Eugenio and Rodari, Giovanni and Francalanci, Tommaso and Suprano, Alessia and Giordani, Taira and Carvacho, Gonzalo and Spagnolo, Nicol{\`o} and Koudia, Seid and Proietti, Massimiliano and Liorni, Carlo and Cerocchi, Filippo and Albiero, Riccardo and Di Giano, Niki and Gardina, Marco and Ceccarelli, Francesco and Corrielli, Giacomo and Chabaud, Ulysse and Osellame, Roberto and Dispenza, Massimiliano and Sciarrino, Fabio},
  year = 2025,
  month = jan,
  journal = {Nature Communications},
  volume = {16},
  number = {1},
  pages = {902},
  issn = {2041-1723},
  doi = {10.1038/s41467-025-55877-z},
  urldate = {2025-12-15},
  langid = {english}
}

@article{yinExperimentalQuantumenhancedKernelbased2025,
  title = {Experimental Quantum-Enhanced Kernel-Based Machine Learning on a Photonic Processor},
  author = {Yin, Zhenghao and Agresti, Iris and De Felice, Giovanni and Brown, Douglas and Toumi, Alexis and Pentangelo, Ciro and Piacentini, Simone and Crespi, Andrea and Ceccarelli, Francesco and Osellame, Roberto and Coecke, Bob and Walther, Philip},
  year = 2025,
  month = sep,
  journal = {Nature Photonics},
  volume = {19},
  number = {9},
  pages = {1020--1027},
  issn = {1749-4885, 1749-4893},
  doi = {10.1038/s41566-025-01682-5},
  urldate = {2025-12-15},
  langid = {english}
}

@article{petersMachineLearningHigh2021,
  title = {Machine Learning of High Dimensional Data on a Noisy Quantum Processor},
  author = {Peters, Evan and Caldeira, Jo{\~a}o and Ho, Alan and Leichenauer, Stefan and Mohseni, Masoud and Neven, Hartmut and Spentzouris, Panagiotis and Strain, Doug and Perdue, Gabriel N.},
  year = 2021,
  month = nov,
  journal = {npj Quantum Information},
  volume = {7},
  number = {1},
  pages = {161},
  issn = {2056-6387},
  doi = {10.1038/s41534-021-00498-9},
  urldate = {2025-12-09},
  langid = {english}
}

@article{schnabelQuantumKernelMethods2025,
  title = {Quantum {{Kernel Methods}} under {{Scrutiny}}: {{A Benchmarking Study}}},
  shorttitle = {Quantum {{Kernel Methods}} under {{Scrutiny}}},
  author = {Schnabel, Jan and Roth, Marco},
  year = 2025,
  month = jun,
  journal = {Quantum Machine Intelligence},
  volume = {7},
  number = {1},
  eprint = {2409.04406},
  primaryclass = {quant-ph},
  pages = {58},
  issn = {2524-4906, 2524-4914},
  doi = {10.1007/s42484-025-00273-5},
  urldate = {2025-12-09},
  archiveprefix = {arXiv}
}

@article{schuldQuantumMachineLearning2019,
  title = {Quantum Machine Learning in Feature {{Hilbert}} Spaces},
  author = {Schuld, Maria and Killoran, Nathan},
  year = 2019,
  month = feb,
  journal = {Physical Review Letters},
  volume = {122},
  number = {4},
  eprint = {1803.07128},
  primaryclass = {quant-ph},
  pages = {040504},
  issn = {0031-9007, 1079-7114},
  doi = {10.1103/PhysRevLett.122.040504},
  urldate = {2025-12-09},
  archiveprefix = {arXiv}
}

@article{thanasilpExponentialConcentrationQuantum2024,
  title = {Exponential Concentration in Quantum Kernel Methods},
  author = {Thanasilp, Supanut and Wang, Samson and Cerezo, M. and Holmes, Zo{\"e}},
  year = 2024,
  month = jun,
  journal = {Nature Communications},
  volume = {15},
  number = {1},
  pages = {5200},
  issn = {2041-1723},
  doi = {10.1038/s41467-024-49287-w},
  urldate = {2025-12-09},
  langid = {english}
}

@article{multicopy-cerf-1,
  title = {Multicopy uncertainty observable inducing a symplectic-invariant uncertainty relation in position and momentum phase space},
  author = {Hertz, Anaelle and Oreshkov, Ognyan and Cerf, Nicolas J.},
  journal = {Phys. Rev. A},
  volume = {100},
  issue = {5},
  pages = {052112},
  numpages = {13},
  year = {2019},
  month = {Nov},
  publisher = {American Physical Society},
  doi = {10.1103/PhysRevA.100.052112},
  url = {https://link.aps.org/doi/10.1103/PhysRevA.100.052112}
}

@article{multicopy-cerf-2,
  title = {Multicopy observables for the detection of optically nonclassical states},
  author = {Arnhem, Matthieu and Griffet, C\'elia and Cerf, Nicolas J.},
  journal = {Phys. Rev. A},
  volume = {106},
  issue = {4},
  pages = {043705},
  numpages = {16},
  year = {2022},
  month = {Oct},
  publisher = {American Physical Society},
  doi = {10.1103/PhysRevA.106.043705},
  url = {https://link.aps.org/doi/10.1103/PhysRevA.106.043705}
}

@article{multicopy-cerf-3,
  title = {Interferometric measurement of the quadrature coherence scale using two replicas of a quantum optical state},
  author = {Griffet, C\'elia and Arnhem, Matthieu and De Bi\`evre, Stephan and Cerf, Nicolas J.},
  journal = {Phys. Rev. A},
  volume = {108},
  issue = {2},
  pages = {023730},
  numpages = {11},
  year = {2023},
  month = {Aug},
  publisher = {American Physical Society},
  doi = {10.1103/PhysRevA.108.023730},
  url = {https://link.aps.org/doi/10.1103/PhysRevA.108.023730}
}

@article{multicopy-cerf-4,
  title = {Accessing continuous-variable entanglement witnesses with multimode spin observables},
  author = {Griffet, C\'elia and Haas, Tobias and Cerf, Nicolas J.},
  journal = {Phys. Rev. A},
  volume = {108},
  issue = {2},
  pages = {022421},
  numpages = {15},
  year = {2023},
  month = {Aug},
  publisher = {American Physical Society},
  doi = {10.1103/PhysRevA.108.022421},
  url = {https://link.aps.org/doi/10.1103/PhysRevA.108.022421}
}

\section{Appendix }
\paragraph{Quasiprobability distributions -} Quasiprobability distributions provide a phase-space representation of quantum states that resembles classical probability theory while accommodating intrinsically quantum features. Unlike true probabilities, they may take negative or even complex values, with these deviations signaling various notions of nonclassicality. This framework offers an intuitive viewpoint on quantum phenomena and is widely used in the analysis of quantum systems \cite{cahillOrderedExpansionsBoson1969,ferrieQuasiprobabilityRepresentationsQuantum2011,hilleryDISTRIBUTIONFUNCTIONSPHYSICS}. Here, we show how our protocol can be used to directly measure the point-wise value of various quasiprobability representations: the Wigner function, the Husimi-Q distribution \cite{curilefHusimiDistributionDevelopment2013}, the positive P representation, and the Kirkwood-Dirac quasiprobability distribution.
The Husimi Q distribution of state $\rho$ is given by 
\begin{equation}
    Q(\alpha)=\frac{1}{\pi}\langle \alpha|\rho|\alpha\rangle=\frac{1}{\pi}\tr{\rho |\alpha\rangle\langle \alpha|},
\end{equation}
It is well known that the Husimi-Q function can be reconstructed by heterodyne measurements. Instead, we note that $Q(\alpha)$ is an overlap, i.e. a second-order Bargmann invariant, which can be estimated using our scheme via the interference of $\rho$ and coherent state $|\alpha\rangle$ at a beam splitter. The expectation value of the photon number parity on output mode $\hat{a}_1$ then gives us an estimate of the overlap, using Eqs. (\ref{eq:P_j}) and (\ref{eq:Pj_vs_Xk}), see also \cite{multicopy-cerf-3}. The same overlap estimation scheme gives a direct measurement also for the Wigner function 
\begin{align}
    W(\alpha) &=\frac{2}{\pi}\tr{\rho D(\alpha)e^{i\pi \hat{n}}D(\alpha)^{\dagger}}\\
    &= \frac{2}{\pi}\tr{\rho D(\alpha)(\sum_{n=0}^{\infty}(-1)^{n}|n\rangle \langle n|)D(\alpha)^{\dagger}}\\
    &= \frac{2}{\pi}\sum_{n=0}^{\infty}(-1)^{n}\tr{\rho D(\alpha)|n\rangle \langle n|D(\alpha)^{\dagger}}  \ \ .
\end{align}
So the Wigner function can be directly measured by collecting the overlap values of the state we are interested in and Fock displaced states $ D(\alpha)|n\rangle \langle n|D(\alpha)^{\dagger}$. Although not optimal, it shows a further connection between Bargmann invariants and quasi-probability distributions.
For an alternative direct measurement scheme for the Wigner function, based on a quantum nondemolition measurement of the displaced parity operator, see \cite{lutterbach97}.

Other quasiprobability distributions can be measured with our scheme in a similar way. For example, the positive P-representation distribution \cite{carmichael2008statistical}, defined as
\begin{equation}
    P(\alpha,\beta)=\frac{\langle \beta|\alpha\rangle \langle \alpha|\rho|\beta\rangle}{\pi^{2}}=\frac{\tr{|\alpha\rangle \langle \alpha|\rho|\beta\rangle\langle \beta|}}{\pi^{2}},
\end{equation}
is given by a third-order Bargmann invariant which is thus measurable with our protocol. Another quasiprobability function of interest is the Kirkwood-Dirac distribution \cite{kirkwoodQuantumStatisticsAlmost1933,diracAnalogyClassicalQuantum1945} defined as
\begin{equation}
    Q_{i,j}(\rho)=\langle b_{j}|a_{i}\rangle\langle a_{i}|\rho|b_{j}\rangle=\tr{|a_{i}\rangle\langle a_{i}|\rho|b_{j}\rangle\langle b_{j}|}
\end{equation}
where $\{|a_{i}\rangle\}$ and $\{|b_{j}\rangle\}$ are two orthonormal bases - this is also clearly a third-order Bargmann invariant. The KD distribution can be extended to include a larger number of bases, or in general to projection-valued measures (PVMs), representable using functions of higher-order Bargmann invariants.

\subsection{Proofs}\label{sec: Proofs} 
Here we present the necessary lemma used in Eq.~\eqref{eq: multivariate_cycle}. Here, we assume the internal d.o.f. $\alpha$ are discrete with $\alpha\in \{1, ...,d\}$ but the result can be extended for the continuous case.
\begin{lem}
    $\tr{\hat{C}\otimes_{i=1}^{M}\rho_{i}}=\tr{\rho_{1}\dots \rho_{M}}$.
\end{lem}
\begin{proof}
    Let us start from pure states written in the form:
\begin{equation*}
    |\phi_{i}\rangle\langle \phi_{i}|=g_{i}(a_{i,1}^{\dagger},\dots,a_{i,d}^{\dagger})|vac\rangle\langle vac|g_{i}(a_{i,1}^{\dagger},\dots,a_{i,d}^{\dagger})
\end{equation*}
where $g_{i}$ are functions of the creation operators, such that $|\phi_{i}\rangle\langle \phi_{i}|$ is normalized. We will use the notation $g_{i}(a_{i,1}^{\dagger},\dots,a_{i,d}^{\dagger})=g_{i}(\boldsymbol{a}_{i}^{\dagger})$, to make the notation more compact. Then we can rewrite the original trace as:
\begin{gather*}
    \operatorname{Tr}(\hat{C}\otimes_{i=1}^{M}|\phi_{i}\rangle\langle \phi_{i}|) =\\
    =\langle vac|{\big (}\prod_{j}g_{j}(\boldsymbol{a}_{i}^{\dagger}){\big )}^{\dagger}\hat{C}{\big (}\prod_{i}g_{i}(\boldsymbol{a}_{i}^{\dagger}){\big )}|vac\rangle\\
    =\langle vac|{\big (}\prod_{j}g_{j}(\boldsymbol{a}_{j}^{\dagger}){\big )}^{\dagger}\hat{C}{\big (}\prod_{i}g_{i}(\boldsymbol{a}_{i}^{\dagger})\hat{C}^{\dagger}C|vac\rangle\\
    =\langle vac|{\big (}\prod_{j}g_{j}(\boldsymbol{a}_{j}^{\dagger}){\big )}^{\dagger}{\big (}\prod_{i}g_{i}(\boldsymbol{a}_{i+1}^{\dagger}){\big )}|vac\rangle\\
    =\prod_{i}\langle vac|g_{i}^{\dagger}(\boldsymbol{a}_{i}^{\dagger})g_{i+1 \operatorname{mod} M}(\boldsymbol{a}_{i}^{\dagger})|vac\rangle\\
    =\prod_{i}\langle \phi_{i}|\phi_{i+1 \operatorname{mod} M}\rangle
    =\operatorname{Tr}(|\phi_{1}\rangle\langle \phi_{1}|...|\phi_{M}\rangle\langle \phi_{M}|)
\end{gather*}
where we reduce the double product to a single one, due to the fact that states in two different external modes are orthogonal. For mixed states, we can notice that through the linearity of the trace, the same argument holds, implying 
\begin{equation*}
    \tr{\hat{C}\otimes_{i=1}^{M} \rho_{i}}=\tr{\rho_{1}\dots\rho_{M}}
\end{equation*}
\end{proof}
To clarify the notation of the proof, let us consider some examples:
\begin{itemize}
    \item for the case of single photons in some internal degree of freedom a common factorization\cite{tichyInterferenceIdenticalParticles2014} can be written as $g_{i}(\boldsymbol{a_{i}}^{\dagger})=a_{i,\phi_{i}}^{\dagger}=\sum_{\alpha}c_{i,\alpha}a_{i,\alpha}^{\dagger}$;
    \item for coherent states we have that $|\boldsymbol{\beta}\rangle=D(\boldsymbol{\beta})|vac\rangle$, where $\boldsymbol{\beta}\in \mathbb{C}^{d}$, in such a case $g_{i}(\boldsymbol{a}_{i}^{\dagger})=e^{-\frac{|\boldsymbol{\beta}|^{2}}{2}}e^{\sum_{\alpha}\beta_{\alpha} a_{i,\alpha}^{\dagger}}$.
\end{itemize}
This approach generalizes the known results for single photon states.\\

Theorem \ref{theo:violation_suplaws} provides a useful tool to understand the impact of the multivariate traces $\{X_{k}\}$ on the suppression laws.
\begin{theorem}
    The probabilities $P_j$ are related to the expectation values $X_k$ via a discrete Fourier transform via the equations
    \begin{equation}\label{eq:PjvsXk}
        P_j = \frac{1}{M} \sum_{k=0}^{M-1}e^{-i \frac{2\pi j k}{M} } X_k 
    \end{equation}
\end{theorem}

\begin{proof}
    \begin{gather*}
        \frac{1}{M} \sum_{k=0}^{M-1}e^{-i \frac{2\pi j k}{M} } X_k = \frac{1}{M} \sum_{k=0}^{M-1}e^{-i \frac{2\pi j k}{M} } \tr{ \hat{D}^{k}\Omega_{out}}=\\
    = \frac{1}{M} \sum_{k=0}^{M-1} \tr{ e^{i \frac{2\pi  k}{M}\left(\sum_{l=0}^{M-1}l\hat{N}_l - j\right)}\Omega_{out} }
    \end{gather*} 
    Since $e^{i \frac{2\pi  k}{M}\left(\sum_{l=0}^{M-1}l\hat{N}_l - j\right)}$ is a diagonal operator in the Fock basis, we can replace the above with
    \begin{gather*}
        \frac{1}{M} \sum_{k=0}^{M-1}e^{-i \frac{2\pi j k}{M} } X_k = \frac{1}{M}\sum_{k=0}^{M-1} \sum_{\vec{S}}p(\vec{S})  \, e^{i \frac{2\pi  k}{M}\left(\sum_{l=0}^{M-1}ls_l - j\right)}\\
        =\sum_{\vec{S}}p(\vec{S})\,\delta_{f(\vec{S}),j}=P_{j}
    \end{gather*}
    where we have used $f(\vec{S})= \sum_{l=0}^{n-1} l s_l \,\operatorname{mod} n$ and the definition of $P_j$. Lastly we notice that, thanks to the properties of the complex exponent, the $\operatorname{mod} n$ operation is native of this function.

Similarly, it can be seen the following
\begin{align}
   P_{j} = \sum_{\vec{S}\sim D_{\vec{S}}}p(\vec{S})\,\delta_{f(\vec{S}),j}= \mathbb{E}[\delta_{f(\vec{S}),j}]  \\
   X_{k} =\sum_{\vec{S}\sim D_{\vec{S}}}p(\vec{S})\,e^{\frac{2\pi i k}{n}f(\vec{S})}= \mathbb{E}[e^{\frac{2\pi i k}{M}f(\vec{S})}] 
\end{align}
and by property of the discrete Fourier Transform, the same can be derived.

\end{proof}

Corollary \ref{cor:sup_laws} characterizes the cyclic invariance in terms of measurement statistics.
\begin{corollary}
    The input state $\Omega$ of the Fourier interferometry process is invariant under the cyclic permutation, i.e. $\hat{C}\Omega=\Omega$, if and only if $P_0=1$. Consequently, any outcome $\vec{S}$ such that $f(\vec{S}) \neq 0$ is forbidden. 
\end{corollary}
\begin{proof}

    The forward implication is a straightforward application of Theorem~\ref{theo:violation_suplaws}. If $\hat{C}\Omega=\Omega$ it is clear that $X_k=1, \forall k \in \{0, .., M-1\}$. Eq.~\eqref{eq:Pj_vs_Xk} then implies that $P_0=1$ and $P_j=0, j\neq 1$. The reverse implication can be prove using the $\hat{\Pi}_{C}$. If $P_0=1$ then $P_j=0, \forall j\neq 1$, since $P_j$ is a probability distribution. This implies that $X_k=1~\forall k$ from Eq.~\eqref{eq:Pj_vs_Xk}. In turn, this implies that $\tr{\hat{\Pi}_{C}\Omega}=1$. Since $\hat{\Pi}_{C}$ is a projector and $\Omega$ is a quantum state, this is only possible if $\hat{\Pi}_{C}\Omega= \Omega$. Finally, since $\hat{C} \hat{\Pi}_{C} = \hat{\Pi}_{C} $, we have that $\hat{C}\Omega= \hat{C} \hat{\Pi}_{C} \Omega = \hat{\Pi}_{C} \Omega =\Omega$. Lastly, since $P_{0}=\mathbb{E}[\delta_{f(\vec{S}),0}]$, it is clear that all the events for $f(\vec{S}) \neq 0$ are suppressed if and only if $P_{0}=1$.
\end{proof}
This implies that $P_{0}$ certifies the cyclic symmetry of $\Omega$, and violation of the suppression laws directly relates to deviations from cyclic invariance of the input state.\\

To prove the sample complexity bound on the values $\tr{\rho_{1}\dots\rho_{M}}$, it is easier to use the relation between $P_{j}$ and $X_{k}$, to avoid dealing with complex values.
\begin{proposition}
    Let $\Omega= \bigotimes_{j} \rho_j$ be a bosonic quantum state where each  $\rho_j$ belongs to a multimode bosonic Fock space $\mathcal{H}$. The multivariate trace $\tr{\rho_1\rho_2...\rho_n}$ can be estimated with probability $1-\delta$ and precision $\epsilon$ with  $O(\epsilon^{-2}\ln\delta^{-1})$ samples.
\end{proposition}
\begin{proof}
Notice that $P_{j}$ is a Bernoulli random variable, either $f(\vec{S})=j$ or not. By using the Hoeffding's inequality, let $Y_{1},\dots,Y_{N}$ be i.i.d. Bernoulli random variables, and let us define: $Z_{N}=\sum_{i=1}^{N}Y_{i}$, then we have that:
\begin{equation}
    Pr{\bigg [}|Z_{N}-E[Z_{N}]|\geq \epsilon{\bigg ]}\leq 2\exp(-2N\epsilon^{2})
\end{equation}
We want that this quantity is smaller than a failure probability $\delta$ , which implies :
\begin{equation}
    N\geq \frac{1}{2\epsilon^{2}}\ln\frac{2}{\delta}
\end{equation} 
Thus $P_{j}$ can be estimated with probability $1-\delta$ and precision $\epsilon$ in $O(\frac{1}{2\epsilon^{2}}\ln(\frac{2}{\delta}))$. If $\{P_{j}\}$ can be estimated up to $\epsilon$ precision, its Fourier Transform $\{X_{k}\}$ yields the same error scaling.
\end{proof}

\newpage

\end{document}